\documentclass{amsart}
\usepackage{amssymb}
\usepackage{epsfig}
\usepackage{latexsym}
\usepackage{amsmath}
\usepackage{mathrsfs}
\newtheorem{theorem}{Theorem}[section]
\newtheorem{proposition}[theorem]{Proposition}

\newtheorem{lemma}[theorem]{Lemma}

\newcommand{\T}{\mathcal T}

\newcommand{\PP}{{\mathbb P}}
\newcommand{\EE}{{\mathbb E}}

\newcommand{\old}[1]{{}}

\title[Resolving deep divergences]{Sequence length bounds for resolving a deep phylogenetic divergence}
\author{Mareike Fischer and Mike Steel*}

\date{\today}



\begin{document}

 \begin{abstract}
In evolutionary biology, genetic sequences carry with them a trace
of the underlying tree that describes their evolution from
a common ancestral sequence.  The question of how many sequence
sites are required to recover this evolutionary relationship
accurately depends on the model of sequence evolution, the
substitution rate, divergence times and the method used to infer
phylogenetic history. A particularly challenging problem for
phylogenetic methods arises when a rapid divergence event occurred
in the distant past.  We analyse an idealised form of this problem
in which the terminal edges of a symmetric four--taxon tree are some
factor ($p$) times the length of the interior edge. We determine an
order $p^2$  lower bound on the growth rate for the sequence length
required to resolve the tree (independent of any particular branch
length). We also show that this rate of sequence length growth can
be achieved by existing methods (including the simple `maximum parsimony'
method), and compare these order $p^2$ bounds with an order $p$
growth rate for a model that describes low-homoplasy evolution.  In
the final section, we provide a generic bound on the 
sequence length requirement for a more general class of Markov processes.
\end{abstract}

\maketitle

\noindent {\em Allan Wilson Centre for Molecular Ecology and Evolution\\
Biomathematics Research Centre,  University of Canterbury \\
Private Bag 4800, Christchurch, New Zealand}

*Corresponding Author:
Phone: +64-3-3667001, Ext. 7688
Fax: +64-3-3642587
Email: m.steel@math.canterbury.ac.nz, email@mareikefischer.de

\newpage

\bigskip
\section{Introduction}
When sequence sites evolve independently under a Markov process
along the branches of a tree $\T$, the sequences observed at the tips
contain information concerning the underlying tree. This allows for the tree
$\T$ to be reconstructed accurately from sufficiently long
sequences; this is the basis of modern molecular systematics
\cite{fels}.  The number of sites required to
reconstruct $\T$ accurately depends on how long the edges of the tree are. More precisely, it depends on the expected number of
substitutions on each branch (edge) $e$ of the tree -- which we refer to as the {\em branch length} of $e$ (this is
the product of the temporal duration of the branch and the
substitution rate).

A number of authors (e.g. \cite{churchill_haeseler_navidi, lec, sai, townsend, wort, xia, yang}) have considered various ways to quantify the phylogenetic
signal in aligned DNA sequences, and to estimate the sequence length required to reconstruct a phylogenetic tree. Most of these studies have involved
simulation or heuristic approaches, although some analytical bounds have also been obtained \cite{mos1, steel_szekely}.  Typically, these bounds state that if an interior branch length is very short, or if a terminal (external) branch length is long, then a large number of sites will be required.

In this paper we explore these results further by obtaining bounds that are expressed purely in terms of the relative sizes of the branch lengths, not their absolute values. One motivation for our approach is that different genes are known to evolve at different rates, so that any particular branch length will depend on which gene is considered; however, the ratios of the branch lengths will be unchanged if the gene-specific rate applies uniformly across the tree.

A particularly difficult tree reconstruction problem, requiring long
sequences to resolve, arises when one has an interior edge with a
short branch length incident with edges (or subtrees) having large
branch lengths. Such a scenario occurs, for example, when a relatively
rapid speciation event (leading to the short branch length for that edge) occurred in the distant past
(leading to the large branch lengths for the incident edges).  Several examples of this have been highlighted in the literature \cite{loc, rok} and include the origin of metazoa  and the
origin of photosynthesis.

In this paper we analyse a scenario which, although somewhat
idealised, nevertheless captures the essence of this problem -- a
four-taxon tree, where the terminal edges have equal branch lengths that
are $p>1$ times the branch lengths of the interior edge, and a
simple symmetric model of site evolution (specifically, we assume
sites evolve independently according to a common two--state Markov process).

We provide a mathematical analysis to the question of how many sites
are required to resolve the tree correctly (from the three possible
resolved topologies on four taxa).  We are particularly interested in
how the growth of the sequence length, $k$, depends on $p$,
independent of the absolute value of a particular edge length.  We
establish that $k$ must grow at the rate $p^2$, which implies that
regardless of how fast (or slow) any particular sequence is
evolving, we can set definite lower bounds on the length of
sequences required to resolve the tree.  We then show that for our
setting, $p^2$ growth in $k$ is the best possible, as an existing method
(namely, maximum parsimony) achieves this bound.  Our results complement an earlier simulation-based analysis  \cite{yang}.  We contrast our results by
considering a quite different model of site evolution (the infinite
state model) and establishing that order $p$ growth in $k$ can
sometimes suffice for this model.

We also extend the approach to more general markov processes on trees, obtaining exact, but less explicit lower bounds on $k$ and which involve absolute (rather than relative) branch lengths. 
Our arguments are based on standard techniques from probability theory,
such as central limit approximation, and information-theoretic arguments based on the properties of Hellinger distance.

\section{Preliminaries}

Consider an unrooted binary phylogenetic tree on four taxa, say
$12|34$, with branch length $x$ for the interior edge $e_5$ and $px$
for the terminal edges $e_1,\ldots,e_4$, where $p >
1$. This is illustrated in Fig. 1(a), and the topology of the
tree is shown at the top of Fig. 1(b). The other two competing
topologies ($13|24$ and $14|23$) are also shown in Fig. 1(b). Here
branch length refers to the expected number of substitutions under
some continuous time substitution process.

\begin{figure}[ht] \begin{center}
\resizebox{8cm}{!}{
\input{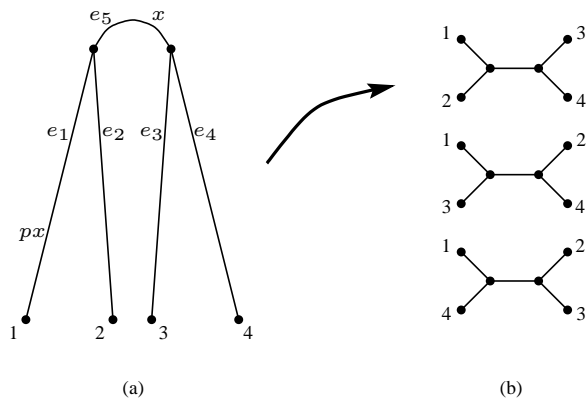}
}
\caption{(a) The generating tree with interior branch length $x$ and all four terminal branch lengths equal to $px$. (b) This tree has the topology $12|34$, while the other two binary topologies are $13|24$ and $14|23$.}
\end{center}
\label{figure1}
\end{figure}

Recall that a {\em binary character} or {\em site pattern} refers to an assignment to each taxon of a state from some two-element set, which we will denote through this paper as $\{\alpha, \beta\}$.

Suppose that a sequence of binary characters are generated independently and identically (i.i.d.) under a symmetric two-state model on the tree.  This model is often called the CFN (Cavender-Farris-Neyman model) or more briefly the Neyman 2-state model (for more details see e.g. \cite{semple_steel}).    Although it is the 
simplest non-trivial Markov process on a tree, it allows for an
exact analysis. Moreover, stochastic results for this model
typically extend to more general finite-state models where an exact
analysis is usually more complex \cite{mos1}, and in
Section~\ref{extend} we show how some of our approaches extend to
more general Markov processes.

If we denote the substitution probability on edge $e_i$ by $P(e_i)$, then for each terminal edge we
have $P(e_i) = \frac{1}{2}(1-2\exp(-2px))$ while for the central
edge $e_5$, we have $P(e_5) = \frac{1}{2}(1-2\exp(-2x))$. Let
$\theta_i=1-2P(e_i)$ for $i =1,\ldots, 5.$ Then we can express
these five $\theta_i$ values in terms of $\theta:=e^{-2x}$ as
follows:
$$\theta_i=\theta^p \mbox { for } i=1,\ldots, 4; \mbox{ and } \theta_5=\theta.$$

Now, if we fix $x$ and let $p$ grow, or, alternatively, if we fix
$px$ and let $x$ tend to zero, then the sequence length $k$
required to reconstruct the topology of the generating tree accurately tends to
infinity.  This holds for any tree reconstruction method that treats
all three topologies fairly (if a method has an a priori preference
for one topology, it will perform worse on an alternative topology).
For example, if $px$ is fixed, then $k$ grows at the rate
$\frac{1}{x^2}$ as $x$ tends to zero (by Theorem 4.1 of
\cite{steel_szekely}).  However, if we do not fix $x$ or $px$ in
advance two fundamental questions arise: what is the slowest rate
that $k$ can possibly grow as a function of $p$? and (ii) does some
value of $x$ (dependent on $p$) achieve this rate of growth for a
certain tree reconstruction method?  We will see that for the simple
scenario described, the answers to these questions are (i) $p^2$ and
(ii) yes (up to a constant factor).

\section{Lower bounds}

The main result of this section is the following:

\begin{theorem} \label{OtherMethodsTheorem}  Suppose $k$ sites evolve i.i.d. under
a symmetric two-state model on some (unknown) four-taxon tree that has
branch length $x$ on the interior edge and $px$ on each terminal edge.
 Then any method that is able to correctly
identify the underlying tree topology with probability at least
$1-\epsilon$ requires:
\begin{equation*} k \geq c_{\epsilon}\cdot p^2 \end{equation*}
for any $x$, where
$c_{\epsilon}= \frac{1}{2}(1-\frac{3}{2}\epsilon)^2$.
\end{theorem}
To establish this result we require some preliminary results. We
begin with a general information-theoretic bound on the number of
i.i.d. observations required to reconstruct a discrete parameter in
a general setting.

Suppose one has a finite set $A$, and each element $a \in A$ has an
associated probability distribution on a finite set $U$. Suppose we
observe $k$ observations from $U$ that are generated independently
by the same unknown element $a \in A$. Suppose, furthermore, that
some method $M$ estimates the element of $A$ that generated our
observations and does so correctly with probability at least
$1-\epsilon$ (regardless of which element $a$ actually generated the
data). Then we can set a lower bound on $k$ in terms of a stochastic
distance between elements of $A$. Recall that the {\em Hellinger
distance} of two elements $a,a' \in A$ is defined as follows. If $p$
and $q$ denote the probability distribution induced by $a$ and $a'$
respectively then let:
\begin{equation} \label{Hellinger}
d_H^2(a,a'):= \sum\limits_{u\in U} \left( \sqrt{p_u}-\sqrt{q_u}\right)^2 = 2\left( 1- \sum\limits_{u\in U}\sqrt{p_uq_u}\right).
\end{equation}
The latter equality holds as $\sum\limits_{u\in U} p_u=
\sum\limits_{u\in U} q_u= 1$. The following result is from
\cite{steel_szekely} (Theorem 3.1 and (2.7)).
\begin{lemma}
\label{infolemma} If there is a subset $A'$ of  $A$ of size $m\geq
2$ for which $d_H(a,a') \leq d$ for all $a,a' \in A'$ and some
method $M$ correctly identifies each element of $A'$ with
probability at least $1-\epsilon$ from $k$ independently-generated
elements in some set $U$, then:
$$k \geq \frac{1}{4}(1-\frac{m}{m-1}\epsilon)^2 d^{-2}.$$
\end{lemma}
In our setting, $A$ will consist of the three binary four-taxon
trees on leaf set $\{1,2,3,4\}$, $U$ will consist of the assignment
of states of the elements of this leaf set, and $m$ will be $3$ (in
this section) or $2$ (in Section~\ref{extend}).

Let $S$ be the set of possible binary site patterns on
$\{1,2,3,4\}$. These consist of the site patterns
$s_1:=\alpha\alpha\beta\beta, s_2:=\alpha\beta\alpha\beta$ and
$s_3:=\alpha\beta\beta\alpha$, and five non-informative ones
$s_4,\ldots,s_8$ (note that pairs of complementary site patterns --
for example  $\alpha\alpha\beta\beta$ and $\beta\beta\alpha\alpha$
-- are regarded as equivalent). For any site pattern $s \in S$, let
$p_{s} = \PP(s|\T_1)$ (respectively $q_s=\PP(s|\T_2)$) be the
probability that the site pattern $s$ is generated on $\T_1$
(respectively $\T_2$). We can express the probabilities $p_{s_1}$
and $p_{s_2}$ in terms of $\theta=e^{-2x}$ by using the Hadamard
representation of \cite{hen} (see \cite{semple_steel}, Section 8.6).
We have:
\begin{equation}
\label{s1}
p_{s_1} = \frac{1}{8} \cdot \left( 1 + 2\cdot \theta^{2p}-4\cdot \theta^{2p+1}+\theta^{4p} \right),
\end{equation}
and:
\begin{equation}
\label{s2}
p_{s_2} = \frac{1}{8} \cdot \left( 1 - 2\cdot \theta^{2p}+\theta^{4p} \right)=\frac{1}{8} \left(1-\theta^{2p} \right)^2.
\end{equation}
To obtain an upper bound on the Hellinger distance for our problem,
we require a further technical lemma.
\begin{lemma}
\label{useful}
Let $\gamma>1$ and let $h(x) = \frac{x^\gamma(1-x)}{(1-x^\gamma)}$. Then the supremum of $h(x)$ for $x$ in the half-open interval $[0,1)$ equals $\frac{1}{\gamma}$.
\end{lemma}
\begin{proof}
Since $\gamma>1$ it can be checked that $h'(x)>0$ for all $x$ in
$(0,1)$, and so $\sup_{x \in [0,1)} h(x) = \lim_{x\uparrow 1} h(x)$.
By L'H$\hat{{\rm o}}$pital's rule, we have $\lim_{x\uparrow 1} h(x) =
\frac{1}{\gamma}$.
\end{proof}

{\em Proof of Theorem \ref{OtherMethodsTheorem}}.

If any method has a probability of at least $1-\epsilon$ of correctly reconstructing each
of the three binary trees on four taxa from i.i.d. sequences of
length $k$ then, by Lemma~\ref{infolemma} with $m=3$ we have:
\begin{equation}
\label{hellyeq2}
k \geq \frac{(1-\frac{3}{2}\epsilon)^2}{4}\cdot d_H^{-2}.
\end{equation}
where $d_H$ is the maximum Hellinger distance between any two of the
three trees.  Now,  if each of the three trees has the $x,px$
combination of branch lengths (for interior, terminal branches,
respectively) then, by symmetry, all three of these pairwise
Hellinger distances are equal. Moreover, we claim that :
\begin{equation}
\label{hellyeqq}
 d_H^{-2} \geq 2p^2.
 \end{equation}
 which together with  (\ref{hellyeq2}) requires $k \geq c_{\epsilon}p^2$ for
the choice of $c_\epsilon$ described.  Thus it remains to establish (\ref{hellyeqq}).

Without loss of generality, $\T_1=12|34$ and $\T_2= 13|24$. Now, for
all $i=3,\ldots,8$, we have $p_{s_i}=q_{s_i}$.  Furthermore,
$p_{s_1}=q_{s_2}$ and $p_{s_2}=q_{s_1}$ as the given trees are
identical except for their leaf labelling. Consequently, Eqn.
(\ref{Hellinger}) can be simplified as follows:
\begin{eqnarray}
d_H^2(\T_1,\T_2)&=&2\left( 1-\sum\limits_{i=1}^8 \sqrt{p_{s_i} q_{s_i}}\right) = 2\left( 1-\sum\limits_{i=3}^8 p_{s_i} - 2\sqrt{p_{s_1}p_{s_2}}\right) \\
&=& 2\left( 1- \left(1-p_{s_1}-p_{s_2}\right) - 2\sqrt{p_{s_1}p_{s_2}}\right)\\
\label{Hellinger2} &=& 2\left(p_{s_1}+p_{s_2} - 2\sqrt{p_{s_1}p_{s_2}}\right)
\end{eqnarray}
Let $\delta = \frac{1}{2}\theta^{2p}(1-\theta)$. Then
$p_{s_1}=p_{s_2}+\delta$, and so Eqn. (\ref{Hellinger2}) can be
re-written as:
\begin{equation}\label{Hellinger_with_epsilon}
d_H^2(\T_1,\T_2)=4 p_{s_2}\left( 1+\frac{\delta}{2p_{s_2}} - \sqrt{1+\frac{\delta}{p_{s_2}}} \right).
\end{equation}
Applying the inequality $\sqrt{1+y} \geq 1 + \frac{y}{2} - \frac{y^2}{4}$, for any $y>0$, to
$y= \frac{\delta}{p_{s_2}}$ in (\ref{Hellinger_with_epsilon}), gives:
$$d_H^2(\T_1,\T_2) \leq \frac{\delta^2}{p_{s_2}} =2\left[\frac{\theta^{2p}(1-\theta)}{1-\theta^{2p}}\right]^2 \leq \frac{1}{2p^2},$$
where the last inequality follows by  invoking
Lemma~\ref{useful} with $\gamma = 2p, x = \theta$. This establishes (\ref{hellyeqq}) and thereby completes the proof of the theorem.

 \hfill $\Box$

\section{An Upper bound: The Performance of Maximum Parsimony} \label{MPsection}

We now show that the lower bound described above is essentially
`best possible' (up to a constant factor) for the given model, as it
can be achieved for a certain choice of $x$ by a simple tree
reconstruction  method, namely Maximum Parsimony (MP). This method selects the tree that requires
the smallest number of substitutions to extend the sequences at the tips of the tree to (ancestral) sequences
at all the interior vertices of the tree (for further background, the reader can consult, for example, \cite{fels} or \cite{semple_steel}).

The probability that MP correctly reconstructs the true tree $12|34$ will be called the
{\it MP reconstruction probability}. In the following theorem, and subsequently, the
notation $c \sim_p C$ indicates that $c/C$ converges to $1$ as $p$ grows.
 Let $f(\epsilon)$ denote the one-sided
$\epsilon$-critical value for the standard normal distribution,
defined by:
\begin{equation*}
f(\epsilon)=z \Leftrightarrow \int\limits_{-\infty}^{z} \frac{1}{\sqrt{2\pi}}e^{-t^2/2}dt = \epsilon.
\end{equation*}

\begin{theorem} \label{MPtheorem}
Suppose $k$ sites evolve i.i.d. under
a symmetric two-state model on some (unknown) four-taxon tree that has
branch length $x$ on the interior edge and $px$ on each terminal edge.
If $k \geq c' p^2
f(\frac{\epsilon}{2})^2$, where $c' \sim_p 4e^{2}$, an interior branch length
$x$ exists for which the MP reconstruction probability is at least
$1-\epsilon$.
\end{theorem}
In order to prove this theorem, some preliminary work is required.
Suppose we generate  a sequence $\mathcal{C}$ of $k$ i.i.d. sites under the symmetric two-state model. Define the random variables $X_i$ and $Y_k$ as follows. Let:
$$
X_i=\begin{cases}
1, & \text {if } i^{th} \text{ character in $\mathcal{C}$ is of the kind } (\alpha, \alpha, \beta, \beta);\\
-1,  & \text {if } i^{th} \text{ character in $\mathcal{C}$ is of the kind } (\alpha, \beta, \alpha, \beta);\\
0,  & \text{else}.
\end{cases}
$$
and let:
\begin{equation*}
Y_k= \sum\limits_{i=1}^{k} X_i.
\end{equation*}
The probability that MP will favour the tree $12|34$ over $13|24$ is
then $\PP(Y_k>0)$. We will exploit the fact that the random
variables $X_i$ are i.i.d., and so $Y_k$ can be approximated for
large $k$ by a normal distribution with a mean $\mu_k$ and a
standard deviation $\sigma_k$. These two parameters can be easily
described (just) in terms of $\theta, p$ and $k$ as follows.
\begin{lemma} \label{lemma1} \mbox{}
\begin{enumerate}
\item $\mu_k = k \cdot \frac{1}{2}\theta^{2p}(1-\theta)$.
\item $\sigma_k^2= k \cdot \frac{1}{4}(1+2\theta^{4p+1} -2\theta^{2p+1}
-\theta^{4p+2})$.
\item $\frac{\mu_k}{\sigma_k} \geq  \sqrt{k}\cdot \theta^{2p}(1-\theta)$.
\end{enumerate}
\end{lemma}
\begin{proof}
Since $X_1, \ldots, X_k$ are independent and take values $+1,0$ and $-1$, we have:
\begin{itemize}
\item[(i)] $\mu_k = k \cdot \left[ \PP(X_1=1) - \PP(X_1=-1)\right]$
\item[(ii)]  $\sigma_k^2= k \cdot \left[ \PP(X_1=1) + \PP(X_1=-1) - \left[ \PP(X_1=1) - \PP(X_1=-1)\right]^2 \right]$
\end{itemize}
Now in the two-state symmetric model and the generating tree in Fig.
1(a), we have:
$$\PP(X_1 = 1) = p_{s_1}, \mbox { and } \PP(X_1=-1) = p_{s_2},$$
where $p_{s_1}, p_{s_2}$ were given above in Eqns. (\ref{s1}) and (\ref{s2}), respectively. Parts (1) and (2) of the lemma now follow by substitution of the expressions for $p_{s_1}, p_{s_2}$ into (i) and (ii) respectively.
For Part (3), note that Parts (1) and (2) imply that
\begin{equation}
\label{eq}
\frac{\mu_k}{\sigma_k} = \sqrt{k}\cdot \frac{N_{\theta}}{D_{\theta}}
\end{equation}
where $N_\theta = \theta^{2p}(1-\theta); D_\theta =
\sqrt{1+2\theta^{4p+1}-2\theta^{2p+1}-\theta^{4p+2})}.$ We now show
that $D_\theta \leq 1$. We have $1+0.5\theta^{2p+1} \geq
\theta^{2p}$ and so $2\theta^{2p+1} ( 1 - \theta^{2p} + 0.5
\theta^{2p+1}) \geq 0$. Consequently $1- 2\theta^{2p+1}( 1 -
\theta^{2p} + 0.5 \theta^{2p+1} ) \leq 1$, which implies that
$D_{\theta}^2 \leq 1$. Part (3) now follows from (\ref{eq}) by the
inequality $D_{\theta}  \leq 1$.
\end{proof}

\begin{proof}[Proof of Theorem~\ref{MPtheorem}]

Note that the MP reconstruction probability is the probability that MP will favour the true tree $12|34$ over both alternative trees on four taxa, namely $13|24$ and $14|23$. Recall that the event of the tree $12|34$ being favoured over $13|24$ can be expressed as $\PP(Y_k> 0)$. The event of $12|34$ being favoured over $14|23$ can be expressed similarly by defining the random variables $\tilde{X}_i$ and $\tilde{Y}_k$ which are analogous to $X_i$ and $Y_k$, using the character $(\alpha, \beta, \beta, \alpha)$ instead of $(\alpha, \beta, \alpha, \beta)$. Then, the MP reconstruction probability can be written as $\PP\left((Y_k > 0) \cap (\tilde{Y}_k >  0) \right)$.
Let:
$$Z_k = \frac{Y_k-\mu_k}{\sigma_k}.$$ Thus, $Z_k$ is the normalised difference of the parsimony score between tree $13|24$ and  $12|34$ for a $k$ i.i.d. characters generated by the tree in Fig. 1(a).
By Lemma~\ref{lemma1}(3) we have
\begin{equation} \label{main_idea}
\PP(Y_k  \leq 0) = \PP(Z_k \leq  -\frac{\mu_k}{\sigma_k}) \leq \PP\left(Z_k \leq -\sqrt{k}\theta^{2p}(1-\theta)\right).
\end{equation}
Now, by symmetry
of the branch length of the generating tree in Fig. 1(a), we have
$\PP(Y_k \leq 0) = \PP(\tilde{Y}_k \leq  0)$. Moreover, by Boole's
inequality: $$\PP\left((Y_k > 0) \cap (\tilde{Y}_k > 0) \right) \geq
1- \PP(Y_k \leq 0) - \PP(\tilde{Y}_k\leq 0),$$ which, combined with
(\ref{main_idea}), furnishes the following inequality for
the MP reconstruction probability:
\begin{equation}
\label{keq2}
 \PP\left((Y_k> 0) \cap (\tilde{Y}_k >0) \right) \geq 1- 2\PP(Y_k\leq 0)  \geq 1 - 2\PP(Z_k \leq -\sqrt{k}\theta^{2p}(1-\theta)).
 \end{equation}
Now, $\theta^{2p}\cdot(1-\theta)$ has a unique local maximum in
$[0,1]$, namely at $\theta':=1-\frac{1}{2p+1}$, at which it takes
the value $\alpha_p/p$, where $\alpha_p=\left(1- \frac{1}{1+2p}
\right)^{2p} \cdot \frac{p}{(1+2p)} \sim_p \frac{1}{2}e^{-1}$.
Moreover, the difference between the distribution of $Z_k$ and a
standard normal distribution tends uniformly to zero as $p$ (and
hence $k$) grows. This follows by applying standard bounds on the
central limit theorem approximation (see, for example, \cite{zahl};
one cannot directly apply the usual form of the central limit
theorem as the distribution of the $X_i$'s is changing with
increasing $p$). Thus we have $\PP(Z_k \leq
-\sqrt{k}\frac{\alpha_p}{p}) \leq \epsilon/2$ provided that $k$
grows at the rate $c'p^2f(\frac{\epsilon}{2})^2$ for $c' \sim_p
4e^2$.

In summary, by (\ref{keq2}), a value for $\theta$ exists,
namely $\theta'=1- \frac{1}{1+2p}$, and thus a value for $P(e_5) =
\frac{1}{2}(1-\theta') =\frac{1}{2(1+2p)} \sim \frac{1}{4p}$ also exists, for
which the MP reconstruction probability is at least $1-\epsilon$.
This completes the proof.
\end{proof}

\subsection{Remarks}
\begin{itemize}
\item
Regarding Theorem~\ref{MPtheorem}, other tree
reconstruction methods have a similar performance to MP
when $k$ grows at the rate $p^2$. Indeed it is possible that such methods
will require shorter sequences, and better statistical properties on trees with different
tree shapes (as MP is statistically inconsistent under some combinations of branch lengths that lie outside those considered in
the scenario of Fig. 1).  We have chosen to consider MP here, because the analysis is relatively straightforward and it suffices to prove the matching lower $p^2$ bound.

\item
One can also derive a (non-asymptotic) form of
Theorem~\ref{MPtheorem} using Azuma's inequality \cite{alon};
however, the constant term in place of $c_\epsilon$ is larger by
a factor of $32$.

\item
The optimal choice of $x$ of (approximately) $\frac{1}{4p}$ for MP has been observed in a slightly different setting by ~\cite{townsend}. 

\item One can ask whether similar $p^2$ bounds on $k$
will apply for more complex models.  We conjecture that for
stationary, reversible, finite-state Markov processes, the
results will be essentially the same for our tree in Fig. 1, up
to a different constant factor $c$.

\item For Markov processes in which the state space is countably
infinite -- and where a substitution is always to a new state
(the `random cluster model' for homoplasy-free evolution, described in \cite{mos}) -- the
situation regarding sequence length requirements is quite
different. In this case, the required sequence length need only
grow at the rate $p$ (not $p^2$), as the following result shows.

\begin{proposition}
Suppose $k$ sites evolve i.i.d. under
a random cluster model  model on some (unknown) four-taxon tree that has
branch length $x$ on the interior edge and $px$ on each terminal edge. Then for a
constant $c'_{\epsilon}$ which depends just on $\epsilon$, the
following holds: If $k \geq c'_{\epsilon}\cdot p$,
an $x$ exists for which the MP reconstruction probability is at least
$1-\epsilon$.
\end{proposition}
\begin{proof}
In the
random cluster model, the probability of a substitution event on an
edge $e$ can be written as $P(e) = 1-\exp(-l)$ where $l$ is the
expected number of changes on the edge (the branch length).  
Now, the random cluster model only generates characters
that are homoplasy-free on the generating tree; thus MP will
return the generating tree from a sequence of characters, provided this tree is the only one on which
those characters are homoplasy-free.  For a tree with topology $12|34$, this will occur
precisely if at least one of the $k$ characters generated assigns taxa $1,2$ a shared state, and taxa $3,4$ a second shared state that is different to that assigned to $1,2$.   The probability $Q$ that any given character generated by the tree in Fig. 1(a) has this property is given by:
$$
 Q=P(e_5)\prod_{i=1}^5(1-P(e_i)) = (1-e^{-x})(1-e^{-px})^4.
$$
 Moreover, if $k \geq \log(\frac{1}{\epsilon})/Q$  then $1-(1-Q)^k \geq 1-\epsilon$ (using the inequality 
 $-\log(1-Q) \geq Q$).  Consequently, MP will correctly reconstruct the generating tree with probability at least $1-\epsilon$ provided that:
 \begin{equation}
 \label{Qeq2}
 k \geq \log(\epsilon^{-1})\cdot (1-e^{-x})^{-1}(1-e^{-px})^{-4}.
 \end{equation}
 Taking $x = 1/4p$  we have $(1-e^{-x})^{-1}(1-e^{-px})^{-4} \sim \frac{1}{4p}(1-e^{-1/4})$, which, in view of  (\ref{Qeq2}),
establishes the result.
\end{proof}

\end{itemize}

\section{Lower bounds for more general models}
\label{extend}
 In this section we derive a lower bound on the sequence length required for tree reconstruction, for a much wider range of Markov processes.
However, unlike the previous sections our bound is expressed  in terms of the
absolute branch lengths (or bounds on these) rather than in terms of
ratios, and it involves constants that  depend on the details of the model.

We first derive a general lemma.  Consider any continuous-time,
stationary and reversible Markov process.  Let ${\mathcal S}$ denote
its state space, and in keeping with earlier
terminology let $S= {\mathcal S}^4$ (thus in previous sections
${\mathcal S} = \{\alpha, \beta\}$). Let $\T_1$ and $\T_2$ be two
topologically distinct four-taxon trees. Suppose that the branch
lengths of $\T_1$ are arbitrary, and that each edge of $\T_2$ has
the corresponding interior or pendant branch length specified by
$\T_1$ (where the pendant edge incident with leaf $i$ in $\T_1$
corresponds to the pendant edge incident with leaf $i$ in $\T_2$).
For $s=(s_1, s_2, s_3, s_4) \in S$, let $p_s$ (respectively $q_s$)
denote the probability of generating $s$ at the tips of $\T_1$
(respectively $\T_2$).  Let $p'_s$ (respectively $q'_s$) denote the
conditional probability of generating $s$ at the tips of $\T_1$
(respectively $\T_2)$  given that a substitution has occurred on the
central edge of $\T_1$ (respectively $\T_2$), and let $D_s :=
q'_s-p'_s$.  Then we have the following result.
\begin{lemma}
\label{helly2}
$$d_H^2(\T_1, \T_2) \leq l^2\cdot  \sum_{s \in S} \frac{D_s^2}{p_s}$$
\end{lemma}
where $l$ denotes the branch length of the interior edge of $\T_1$.
\begin{proof}
Let $\tau$ denote the probability that at least one substitution
occurs on the interior edge of  $\T_1$, and let $p^0_s$
(respectively $q^0_s$)  denote the conditional probability of
generating $s$ on $\T_1$ (respectively $\T_2$) given that no
substitution occurs on the interior edge of $\T_1$ (respectively
$\T_2$). By the law of total probability we have:
$$p_s = (1-\tau)\cdot p^0_s + \tau\cdot p_s'$$
and
$$q_s = (1-\tau)\cdot q^0_s + \tau\cdot q_s'.$$
Moreover, the assumptions on the correspondence between branch lengths of
$\T_1$ and $\T_2$ imply that $p^0_s = q^0_s$ for all $s\in S$ and so:
$$q_s - p_s = \tau(q_s'- p_s') = \tau D_s.$$
Now, $$d_H^2(\T_1, \T_2) = 2(1-\sum_{s \in S}\sqrt{p_sq_s}) =
2\left(1-\sum_{s \in S} p_s\sqrt{1+\frac{\tau D_s}{p_s}}\right).$$
Applying the inequality $\sqrt{1+y} \geq 1+\frac{y}{2}
-\frac{y^2}{2}$ (for all $y \geq -1$) to $y = \frac{\tau D_s}{p_s}$ (and
observing that $y \geq -1$ since $q_s \geq 0$), we obtain:
$$d_H^2(\T_1, \T_2) \leq 2\left(1-\sum_s p_s\left(1 +\tau\frac{D_s}{2p_s} -
\tau^2\frac{D_s^2}{2p_s}\right)\right).$$ Now, $\sum_s p_s=1$, and $\sum_s D_s =0$
(since $\sum_s q'_s = \sum_s p'_s = 1$) and so this last inequality
reduces to: \begin{equation} \label{yyy} d_H^2(\T_1, \T_2) \leq
\tau^2\cdot \sum_{s \in S} \frac{D_s^2}{p_s}. \end{equation}
Furthermore, $\tau =\PP(N>0)$, where $N$ is the number of substitutions
occurring on the interior edge of $\T_1$. However, $\PP(N>0) \leq
\EE(N)$; that is, $\tau \leq l$, which, together with (\ref{yyy}),
provides the inequality stated in the lemma.
\end{proof}

We now apply this lemma to a slightly more restricted class of
Markov processes to obtain the main result of this section.

\begin{theorem} \label{SecondTheorem}  Suppose $k$ sites evolve i.i.d. under
a finite-state, stationary and reversible continuous-time Markov
process in which each state is accessible from any other state.  Let $l_0$ be any strictly positive value.
Consider this process on some (unknown) four-taxon tree that has
branch length at most $l$ on the interior edge and at least $L\geq
l_0$ on each terminal edge. Then any method that is able to
correctly identify with probability at least $1-\epsilon$ the
underlying tree topologies given these restriction requires:
\begin{equation*} k \geq \frac{C}{4}(1-2\epsilon)^2  \cdot \frac{e^{cL}}{l^2}
\end{equation*}
where $c$ and $C$ are positive constants that depend only on $R$ (the rate matrix for the
process) and $l_0$.
\end{theorem}
\begin{proof}
 We exploit the fact that any Markov process of the type described
converges to its unique stationary distribution at an exponential
rate (see, for example, Theorem 8.3 of \cite{roz}).  Let $\pi(s)$
denote the stationary probability of $s$ under the model.
 For $j=1,\ldots, 4$, let
$p(j) \in \{u,v\}$ be the end of the interior edge $uv$ of $\T_1$
that is adjacent to leaf $j$  (we may assume $p(1)=p(2)=u;
p(3)=p(4)=v$),  and let $S_{p(j)}$ denote the random state present
at that vertex under the model.  Then for any $s_j, s_j' \in
{\mathcal S}$ there exist positive constants $A,a$ (dependent on
$R$) for which:
\begin{equation}
\label{boundy1}
 |\PP(S_j = s_j|S_{p(j)} = s'_j) - \pi(s_j)| \leq A e^{-aL_j}
 \end{equation}
 (\cite{roz}, Theorem 8.3), where $L_j$ denotes the branch length of the edge incident with leaf $j$.
For $s = (s_1, s_2, s_3, s_4) \in S={\mathcal S}^4$, let $$\pi_s =
\prod_{j=1}^4\pi(s_j).$$ For $s' s'' \in {\mathcal S}$ let
$p'(s',s'')$ denote the probability of generating state $s'$ at $u$
and the state $s''$ at $v$ given that at least one substitution
occurs on the edge $uv$. Then, by the Markov assumption, and
recalling the definition of $p_s'$ from Lemma~\ref{helly2}, we have:
\begin{equation}
\label{eqxy}
p_s' = \sum_{(s',s'') \in {\mathcal S}^2} p'(s',s'')\cdot \prod_{j=1}^2\PP(S_j = s_j|S_u = s') \cdot \prod_{j=3}^4\PP(S_j=s_j|S_v=s'').
\end{equation}
Combining (\ref{boundy1}) and (\ref{eqxy}), there exist positive
constants $B,b$ (dependent only on $R$) such that:
\begin{equation}
\label{p1eq1}
|p'_s - \pi_s| \leq B e^{-bL}
\end{equation}
 for all $s \in S$ (recall that $L \leq L_j$ for all $j$).
Now, consider tree $\T_2$ which has branch lengths that correspond
to those in $\T_1$ (as in Lemma~\ref{helly2}). Then we also have:
\begin{equation}
\label{q1eq1}
|q'_s - \pi_s| \leq B e^{-bL}
\end{equation}
 for all $s \in S$. 
 Combining (\ref{p1eq1})
and (\ref{q1eq1}) using the triangle inequality gives:
\begin{equation}
\label{triangle}
|D_s| =|q_s-p_s| \leq 2Be^{-bL}.
\end{equation}
Moreover, since $L_j \geq l_0$ (for all $j$) and each state is
accessible from any other state, we have $p_s \geq \delta$ (for some
$\delta>0$ dependent only on $R$ and $l_0$).  Combining this with
(\ref{triangle}) gives the following inequality, for all $s \in S$:
\begin{equation}
\label{xxx}
\frac{D_s^2}{p_s} \leq (4B^2/\delta) e^{-2bL}.
\end{equation}
The theorem now follows from Lemma~\ref{helly2} and Lemma~\ref{infolemma} (with $m=2$).
\end{proof}

\

\section{Concluding remarks}
In this paper we have provided precise results for a specific and simple model (the two-state symmetric process), along with less explicit results for more general Markov processes (and phrased in terms of    absolute rather than relative branch lengths).  The aim is to determine rigorous bounds on the sequence length required for resolving  a deep divergence, which may she light on debates as to whether some early radiations might be fundamentally unresolvable on the basis of current models and data.

 Of course, in applications, other phenomena (such as lineage sorting, misalignment of sequences, sequencing errors and so forth) may further impede phylogenetic reconstruction (including substitution model mis-specification, lineage sorting and alignment artifacts \cite{phil}),  however these errors are unlikely to help tree reconstruction if our bound shows it is impossible even when the idea model assumptions hold. We have seen that some models require significantly fewer characters for resolving a tree -- in particular  this holds for the random cluster model, and it is possible that new types of genomic data (involving rare genomic events where homoplasy is unlikely) can be described by these and related processes that preserve more phylogenetic signal regarding distant evolutionary divergences.  

One limitation concerning our bounds is that they apply to pure Markov processes, in which each character evolves according to the same process.   In molecular biology  a common assumption is that  there is a distribution of rates across sites, in which each sites evolves at a rate (selected i.i.d. from some distribution) that acts as a multiplier for all the branch lengths in the tree (see e.g. \cite{fels, semple_steel}).   It would be interesting to extend the analysis in the last section to these models to obtain a lower bound on $k$ analogous to Theorem~\ref{SecondTheorem}.

\section{Acknowledgements}
We thank the {\em Allan Wilson Centre for Molecular Ecology and Evolution} for funding this work.

\end{document}